\documentclass[lettersize,journal]{IEEEtran}
\IEEEoverridecommandlockouts
\usepackage[dvipdf]{graphicx,color}
\usepackage{amssymb}
\usepackage{enumerate}
\usepackage{amsmath}
\usepackage{amsfonts}
\usepackage{balance}
\usepackage{multirow}
\usepackage{makecell}
\usepackage{color}
\usepackage{algorithm}
\usepackage{stfloats}
\usepackage{float}
\usepackage{algpseudocode}
\usepackage{color}
\usepackage{varwidth}
\usepackage{multicol}
\usepackage{subfigure}
\usepackage{xspace}
\usepackage{enumerate}
\usepackage{xcolor,cite,etoolbox}
\usepackage{bm}
\usepackage{amsthm}
\usepackage{booktabs}

\usepackage{multicol} 
\usepackage{array}
\newcolumntype{M}[1]{>{\raggedright\arraybackslash}m{#1}}
\newtheorem{proposition}{Proposition}
\newtheorem{corollary}{Corollary}
\newtheorem{remark}{Remark}
\usepackage[font=footnotesize,labelfont=scriptsize]{caption}
\usepackage{ragged2e}

\def\BibTeX{{\rm B\kern-.05em{\sc i\kern-.025em b}\kern-.08em
    T\kern-.1667em\lower.7ex\hbox{E}\kern-.125emX}}

\usepackage{titlesec}
\usepackage{tikz}

\newcommand{\deltaInSqcup}{
  \mathbin{
    \begin{tikzpicture}[baseline=(cup.base)]
      \node (cup) at (0,0) {$\sqcup$};
      \node at (cup.center) {\scriptsize$\delta$};
    \end{tikzpicture}
  }
}
\usepackage[margin=0.64in]{geometry}
\makeatletter 
\pretocmd\@bibitem{\color{black}\csname keycolor#1\endcsname}{}{\fail}
\newcommand\citecolor[1]{\@namedef{keycolor#1}{ \color{red}}}
\makeatother

\begin{document}

 {\title{\huge  Rethinking Passive RIS: Finite Blocklength Reliability Analysis Under Thermal Noise}}	
	\author{{Farjam Karim, \textit{Graduate Student Member}, \textit{IEEE}, Deepak Kumar, \textit{Member}, \textit{IEEE}, Prathapasinghe Dharmawansa, \\ Nurul Huda Mahmood, \textit{Member}, \textit{IEEE},  Arthur Sousa de Sena, \textit{Member}, \textit{IEEE},   Matti Latva-aho, \textit{Fellow}, \textit{IEEE}}\\
\thanks{\hrulefill}
		\thanks{
All the authors except D. Kumar and A.S de Sena are with CWC, University of Oulu, Finland. (e-mail: \{farjam.karim, Prathapasinghe.KaluwaDevage, nurulhuda.mahmood, matti.latva-aho\}@oulu.fi.).\\ 
 D. Kumar is with Department of Electronics and Communication Engineering, Motilal Nehru National Institute of Technology Allahabad, Prayagraj, 211004, India (email: dkumar@mnnit.ac.in).\\
 A.S de Sena is with Ericsson, Sweden. (e-mail: arthurssena@ieee.org).
 
} 

 }

	\maketitle

\begin{abstract}
Short-packet communication alters the fundamental performance limits of reconfigurable intelligent surface (RIS)-assisted systems, making conventional  analyses based on the infinite blocklength regime insufficient. This work investigates RIS-assisted transmission in the finite blocklength (FBL) regime while explicitly incorporating thermal noise generated by passive RIS elements, an effect commonly neglected in existing models. A unified analytical framework is developed to characterize the block-error rate (BLER), its asymptotic behavior, and the resulting goodput under both uniform and non-uniform RIS reflection coefficients. Our results show that ignoring RIS thermal noise leads to a pronounced overestimation of reliability with the mismatch increasing as the number of reflecting elements grows. Furthermore, increasing the RIS size does not always improve performance, particularly in the low transmit power regime where accumulated noise becomes dominant. Overall, the results highlight fundamental limitations of idealized RIS models and demonstrate the need for incorporating thermal noise for accurate system evaluation.
	\end{abstract}
 
\begin{IEEEkeywords}
Block-error rate, finite blocklength, Reconfigurable intelligent surfaces, thermal noise.
\end{IEEEkeywords}

\section{Introduction}
Sixth-generation (6G) wireless networks target massive connectivity with stringent requirements on capacity, latency, and energy efficiency. However, such requirements lead to challenges such as spectrum scarcity and increased energy consumption.

Reconfigurable intelligent surfaces (RIS) are regarded as a key enabling technology for future wireless networks due to their ability to intelligently reconfigure the wireless propagation environment with low energy consumption~\cite{wida_tvt_2026}. Through programmable reflecting elements, RIS can improve communication reliability and coverage without requiring active signal amplification. Despite recent developments in advanced RIS architectures~\cite{anastant_tvt_2024}, the majority of existing literature considers passive RIS under idealized noise-free assumptions (see, e.g., Section II-A of~\cite{Zhang_tcom_23} and references therein). This simplification, however, may result in overly optimistic performance, especially when the RIS is deployed close to the receiver.

In our previous work~\cite{my_paper}, we addressed this limitation by developing a physically grounded model for passive RIS thermal noise, linking the noise power to fundamental physical constants exploiting Johnson-Nyquist theorem~\cite{Nyquist}. Based on this model, closed-form expressions were derived in~\cite{my_paper}, demonstrating that RIS thermal noise can significantly impact system performance.
Nevertheless, the analysis was limited to infinite blocklength transmissions based on Shannon-capacity metrics. However, many emerging applications, such as ultra-reliable low-latency communications, internet-of-things networks, and short-packet 6G systems, inherently operate in the finite blocklength (FBL) regime, where conventional asymptotic performance metrics may no longer provide accurate reliability characterization~\cite{Polyanski_channelcoding_10May,kumari_tvt_2026}. This raises an important question regarding the impact of RIS thermal noise under short-packet transmission constraints. Motivated by this, we develop a thermal-noise-aware RIS-assisted communication framework in the FBL regime to characterize the performance of practical short-packet RIS-assisted systems.
The main contributions of this work are summarized as follows:
\begin{itemize}
\item We develop a thermal-noise-aware RIS-assisted communication framework under the FBL regime, capturing the joint impact of RIS-induced noise and short-packet transmission constraints.
\item We derive analytical expressions for key performance metrics, including BLER, goodput, and asymptotic BLER over Nakagami-$m$ fading channels.
\item We generalize the RIS reflection model by considering both uniform and non-uniform configurations, enabling a more flexible and realistic system design.
\item Numerical results validate the accuracy of the analytical expressions and highlight the impact of RIS thermal noise in FBL communications.
\end{itemize}

%

\vspace{-1em}
\subsection{Notations:}
Scalars are represented by italic letters, while vectors are written in boldface. For a random variable $X$, $f_X(x)$ and $F_X(x)$ denote its probability density function (PDF) and cumulative distribution function (CDF), respectively. The Nakagami-$m$ fading model distribution is characterized by a shape parameter $m_{(\cdot)}$, and an average fading power parameter $\Omega_{(\cdot)}$. A circularly symmetric complex Gaussian variable with zero mean and variance $\sigma^2_{(\cdot)}$ is denoted by
 $\mathcal{CN}(0, \sigma^2_{(\cdot)})$, and $\exp(\cdot)$ represents the exponential function.
The Meijer-G function is expressed as
$ G_{m,n}^{p,q} \left( x \ \middle| \begin{matrix} 
a_1, \dots, a_n, a_{n+1}, \dots, a_p \\
b_1, \dots, b_m, b_{m+1}, \dots, b_q
\end{matrix} \right)$
as defined in~\cite[eqn. $9.30$]{Book1}. The $N \times N$ identity matrix is denoted by  $\mathbf{I}_N$ and ${(\cdot)}^\mathrm{H}$ represents the Hermitian (conjugate transpose) of a vector or matrix. The notation $\deltaInSqcup$ indicates that the positive integer value of $\delta$ is considered. The operators $\mathbb{E}\{\cdot\}$ and $\mathrm{Pr}\{\cdot\}$ denote expectation and probability, respectively. Finally, $\gamma(\cdot, \cdot)$ and $\Gamma(\cdot)$ represent the lower incomplete Gamma function and the Euler Gamma function.

\section{System Model}\label{sec:sysmtem}
Consider a wireless communication system aided by a RIS where a single-antenna transmitter sends information to a single-antenna receiver via an RIS composed of $N$ passive reflecting elements. The direct link between the transmitter and the receiver is assumed to be obstructed, so communication relies entirely on the RIS-reflected path.
Each element of the RIS applies an adjustable phase shift to the incoming signal before reflecting it toward the receiver. Denoting the phase adjustment of the $n$-th element by $\theta_n$, where $\theta_n \in (0, 2\pi]$, the collective effect of all RIS elements can be expressed by the diagonal matrix
$
\mathbf{\Theta} = \mathrm{diag}\left(\sqrt{\beta_1}e^{j\theta_1},\sqrt{\beta_2}e^{j\theta_2}, \ldots, \sqrt{\beta_N}e^{j\theta_N}\right).
$
Note that, the reflected signal experiences scaling by the $n$-{th} element of the RIS with a reflection coefficient denoted by $\beta_n$, with $0 \leq \beta_n \leq 1$.  We first consider a uniform reflection model with $\beta_1=\cdots=\beta_N=\beta$, and subsequently extend the analysis to the more general non-uniform case when defining the signal-to-interference-plus-noise ratio (SINR),  where the reflection coefficients may vary across RIS elements.

The channels corresponding to the transmitter-RIS and RIS-receiver links are modeled as independent and identically distributed Nakagami-$m$ fading channels. It is further assumed that perfect channel state information (CSI) is available at the transmitter, the RIS controller, and the receiver{\footnote{ This assumption serves as an idealized benchmark, while the analysis under imperfect CSI is left for future work.}}.
Thermal noise in the system follows the Johnson-Nyquist principle, where the noise power over operating bandwidth $B$ is given by $k\mathcal{T} B$, with $k$ denoting Boltzmann's constant and $\mathcal{T}$ the absolute temperature~\cite{Nyquist}. At the receiver, this is further scaled by the noise figure. Therefore, the noise generated by the receiver can be calculated as $k \mathcal{T} B \lambda$, where $\lambda$ represents the noise figure from the active circuitary of the receiver. Similarly, each passive RIS element generates thermal noise due to ohmic losses, and the aggregate contribution scales linearly with the number of reflecting elements. Accounting for the uniform $\beta_n$ from every RIS element{\footnote{ For non-uniform $\beta_n$ total RIS-generated thermal noise power can be calculated as $k\mathcal{T}B\sum\limits^{N}_{n=1}\beta_n = \sigma^2_r\sum\limits^{N}_{n=1}\beta_n$ .}}, the total RIS-induced noise power can be expressed in a compact form of $N\beta k\mathcal{T}B$. 
A detailed explanation of this noise model, is provided in our previous work~\cite{my_paper}. Since, the underlying assumptions remain unchanged, we omit the full discussion here for brevity and directly present the received signal expression as
\begin{align}\label{received_sig}
y_{k} = \sqrt{P }\left(\mathbf{h}^{\mathrm{H}}_{d}\boldsymbol{\Theta}\mathbf{h}_{b}\right)x+\mathbf{h}^{\mathrm{H}}_{d}\boldsymbol{\Theta}\boldsymbol{\eta_r} +\! \eta_{d},
\end{align}
where $\mathbf{h}_{d} = [h_{1d}, h_{2d}, \dots, h_{Nd}]^T$, with $h_{nd}$ denoting the channel gain between the $n$-th element of the RIS and the receiver; $\mathbf{h}_{b} = [h_{b1}, h_{b2}, \dots, h_{bN}]^T$, where $h_{bn}$ denotes the channel gain between the BS and the $n$-th RIS element; $\boldsymbol{\eta_r} \sim \mathcal{CN}(0, \sigma^2_r \mathbf{I}_N)$ denotes the thermal noise introduced by the RIS elements, $\eta_{d} \sim \mathcal{CN}(0, \sigma^2_d)$ is the thermal noise at the receiver. Both the noises generated by the RIS and the receiver noise are modeled as additive white Gaussian noise (AWGN) distributions. We set $\sigma^2_d= k \mathcal{T} B \lambda$ and $\sigma^2_r=k \mathcal{T} B$. The transmitted symbol $x$ has unit-power and $P$ denotes the transmit power of the transmitter. It is worth noting that the all of existing works neglect the effect of passive noise re-radiation through the RIS while doing theoretical analysis, particularly in the context of FBL transmission. Based on \eqref{received_sig}, the resulting SINR, assuming a uniform reflection coefficient $\beta_n$ across all RIS elements, can be expressed as
\begin{align}\label{SINR}
\gamma_d = \frac{{\rho} \left|\mathbf{h}^{\mathrm{H}}_{d} \boldsymbol{\Theta} \mathbf{h}_{b}\right|^2}{||\mathbf{h}^{\mathrm{H}}_{d}||^2 \psi + 1},
\end{align}
where $\rho ={P}/{\sigma^2_d}$ is the transmit signal-to-noise ratio (SNR) and  $\psi = {\beta\sigma^2_r}/{\sigma^2_d}$. While \eqref{SINR} assumes a uniform reflection coefficient across all RIS elements, a more practical setting allows element-wise variations. As $\beta_n$ represent the reflection coefficient associated with the $n$-th RIS element. In this case, the SINR expression can be reformulated as

\begin{align}\label{non_unfirom_SINR}
  \tilde{\gamma}_d = \frac{
    \rho\left|\sum_{n=1}^{N}\sqrt{\beta_n}\, h_{n,d}^{*} h_{b,n}\, e^{j\theta_n}\right|^2
    }{ 1+\frac{\sigma_r^2}{\sigma_d^2}\sum_{n=1}^{N}\beta_n|h_{n,d}|^2}.
\end{align}
\section{Performance Evaluation}\label{sec_3}
This section evaluates the performance of the considered system by deriving analytical expressions for the BLER, the corresponding goodput, and the asymptotic behavior of the BLER. The BLER at the receiver for decoding a transmitted message over $\Xi$ channel uses with $\vartheta$ bits per slot, when $\Xi > 100$, can be approximated as~\cite{yang_jiang_tvt_2024}
\begin{align}\label{Qfunc}
\varphi \approx \mathbb{E}\left\{Q\left(\frac{C(\gamma_d)-r}{\sqrt{V(\gamma_d)/\Xi}}\right)\right\},
\end{align}
where $C(\gamma_d)=\log_2(1+\gamma_d)$ denotes the Shannon capacity, $r=\vartheta/\Xi$ represents the transmission rate. The term $V(\gamma_d)$ represents the channel dispersion, which characterizes the stochastic variability of the channel relative to a deterministic channel with the same capacity and it is given by $\left(1-(1+\gamma_d)^{\text{-}2}\right)(\text{log}_2e)^2$. 
Using \eqref{Qfunc}, we evaluate the BLER in the following proposition.
\begin{proposition}
    The approximated BLER at the receiving device considering uniform RIS elements  reflection coefficient can be expressed as
    \begin{align}\label{out_final}
       \varphi \approx \varpi\sqrt{\Xi}\left(\varepsilon_2-\varepsilon_1\right) \left\{1-\big[\left(1-\upsilon_1\right)\left(1-\upsilon_2\right)\big]\right\},
    \end{align}
   where \begin{align}
        &\upsilon_1= 1- \sum\limits^{\deltaInSqcup\!- 1}_{i=0} \frac{1}{i!\sqrt{\pi}} \left(\frac{m_{nd}}{\Omega_{nd}}\right)^{-\frac{i}{2}}\left(\frac{1}{\zeta}\sqrt{\frac{\frac{\left(\varepsilon_2+\varepsilon_1\right)}{2}\psi}{\rho\beta}}\right)^i\nonumber\\
        &\times \frac{1}{\Gamma(m_{nd}N)}G_{1,2}^{2,1} \left( \frac{\psi\Omega_{nd}\frac{\left(\varepsilon_2+\varepsilon_1\right)}{2}}{4\rho\beta m_{nd}\zeta^2} \Bigg| \begin{array}{c} 1-m_{nd}N-\frac{i}{2} \\0, \frac{1}{2} \end{array} \right)\nonumber,\\
      &\upsilon_2=  \frac{1}{\Gamma(\delta)} \, G^{1,1}_{1,2} \left( \frac{1}{\zeta} \sqrt{\frac{\frac{\left(\varepsilon_2+\varepsilon_1\right)}{2}}{\rho\beta}} \;\middle|\; \begin{array}{c} 1 \\ \delta, 0
\end{array} \right)\nonumber; \text{such that}\; \Lambda= 2^{r}-1,
    \end{align} $\mu= \frac{\Gamma\left(m_{bn}+0.5\right)\Gamma\left(m_{nd}+0.5\right)}{\Gamma\left(m_{bn}\right)\Gamma\left(m_{nd}\right)}\left(\frac{\Omega_{bn}\Omega_{nd}}{m_{bn}m_{nd}}\right)^{\frac{1}{2}}N$, $\varpi = \frac{1}{2\pi\sqrt{2^{2r}-1}}$, $\sigma^2= N{\Omega_{bn}\Omega_{nd}}\Bigl[1-\frac{1}{m_{bn}m_{nd}} {\left(\frac{\Gamma\left(m_{bn}+0.5\right)\Gamma\left(m_{nd}+0.5\right)}{\Gamma\left(m_{bn}\right)\Gamma\left(m_{nd}\right)}\right)^2}\Bigl]$,    $\varepsilon_1 = \Lambda - \frac{1}{2\varpi\sqrt{\Xi}}$, $\varepsilon_2 = \Lambda + \frac{1}{2\varpi\sqrt{\Xi}}$,   $\delta= \frac{\mu^2}{\sigma^2}$, and $\zeta=\frac{\sigma^2}{\mu}$.
\end{proposition}
\begin{proof}
    Please refer to Appendix A.
\end{proof}


 If thermal noise is neglected, then the approximated BLER expression can be obtained by putting $\upsilon_1=0$ in \eqref{out_final}.
 
     We now proceed to evaluate the BLER under the more general setting of non-uniform reflection coefficients across the RIS elements considering \eqref{non_unfirom_SINR}, as presented in the following proposition.
\begin{proposition}
     The approximated BLER at the receiving device considering non-uniform RIS elements  reflection coefficient can be expressed as
    \begin{align}\label{non-uniform}
       \hat{\varphi} \approx \varpi\sqrt{\Xi}\left(\varepsilon_2-\varepsilon_1\right) \left\{1-\big[\left(1-\nu_1\right)\left(1-\nu_2\right)\big]\right\},
    \end{align}
   where  \begin{align}
        \nu_1= 1- \sum\limits^{\hat{\deltaInSqcup}- 1}_{i=0} \frac{1}{i!\sqrt{\pi}} \left(\frac{1}{\zeta_{c}}\right)^{-\frac{i}{2}}\left(\frac{1}{\zeta_1}\sqrt{\frac{\frac{\left(\varepsilon_2+\varepsilon_1\right)}{2}\hat{\psi}}{\rho}}\right)^i\nonumber\\
        \times \frac{1}{\Gamma(\delta_c)}G_{1,2}^{2,1} \left( \frac{\zeta_{c}\frac{\left(\varepsilon_2+\varepsilon_1\right)}{2}\hat{\psi}}{4\rho\zeta_1^2} \Bigg| \begin{array}{c} 1-\delta_c-\frac{i}{2} \\0, \frac{1}{2} \end{array} \right)\nonumber,
    \end{align}  \begin{align}
      \nu_2=  \frac{1}{\Gamma(\hat{\delta})} \, G^{1,1}_{1,2} \left( \frac{1}{\zeta_1} \sqrt{\frac{\frac{\left(\varepsilon_2+\varepsilon_1\right)}{2}}{\rho}} \;\middle|\; \begin{array}{c} 1 \\ \hat{\delta}, 0
\end{array} \right)\nonumber; \text{such that}\; \zeta_1=\frac{\hat{\sigma}^2}{\hat{\mu}},
    \end{align}
     $ \hat{\mu} = \sum\limits^{N}_{n=1}{\beta_n} \frac{\Gamma\left(m_{bn}+0.5\right)\Gamma\left(m_{nd}+0.5\right)}{\Gamma\left(m_{bn}\right)\Gamma\left(m_{nd}\right)}\left(\frac{\Omega_{bn}\Omega_{nd}}{m_{bn}m_{nd}}\right)^{\frac{1}{2}} $,  $ \hat{\sigma}^2 = \sum\limits^{N}_{n=1}{\beta^2_n}{\Omega_{bn}\Omega_{nd}}\Bigl[1-\frac{1}{m_{bn}m_{nd}} {\left(\frac{\Gamma\left(m_{bn}+0.5\right)\Gamma\left(m_{nd}+0.5\right)}{\Gamma\left(m_{bn}\right)\Gamma\left(m_{nd}\right)}\right)^2}\Bigl]$,   $\hat{\delta}=\frac{\hat{\mu}}{\hat{\sigma}^2}$, $\hat{\psi}=\frac{\sigma_r^2}{\sigma_d^2}\sum_{n=1}^{N}\beta_n$,  $  \mu_c = \sum\limits^{N}_{n=1}\beta_n\Omega_{nd}$, 
    $\sigma^2_c =\sum\limits^{N}_{n=1}\beta^2_n\frac{\Omega^2_{nd}}{m_{nd}}$, 
$\delta_c=\frac{\mu^2_c}{\sigma^2_c}$, and $\zeta_c=\frac{\sigma^2_c}{\mu_c}$.
\end{proposition}

\begin{proof}
    The derivation follows the same steps as Appendix A from \eqref{qexpand} to \eqref{int_upsilon}, with only parameter modifications. However, $f_Z(z)$ must be evaluated to account for non-uniform $\beta_n$. Specifically, letting $Z\triangleq \sum_{n=1}^{N}\beta_n|h_{nd}|^2$, the PDF of $Z$ can be approximated via moment matching as $f_{Z}(z)=  \frac{z^{\delta_c-1}}{\Gamma(\delta_c)\zeta_c^{\delta_c}}\exp{\left(-\frac{z}{\zeta_c}\right)}.$
Following the same procedure as in Appendix A after \eqref{int_upsilon}, we obtained \eqref{non-uniform}.
\end{proof}
If thermal noise is neglected, then the approximated BLER expression for non-uniform $\beta_n$ can be obtained by putting $\nu_1=0$ in \eqref{non-uniform}.

Due to the involvement of special functions (e.g., the Meijer-G function) in \eqref{out_final} and \eqref{non-uniform}, deriving direct insights can be challenging. To address this, we next develop an asymptotic characterization of the BLER, which provides a better understanding of the system behavior considering uniform and non-uniform reflection coefficients as stated in the following propositions, respectively.
\begin{proposition}
    The asymptotic BLER at the receiving device considering uniform reflection coefficient among all the RIS element can be approximated as
    \begin{align}\label{out_final_asymptoic}
       \varphi^{\infty} \approx \varpi\sqrt{\Xi}\left(\varepsilon_2-\varepsilon_1\right)\left\{1- \left[\left(1-\upsilon_1^{\infty}\right) \left(1-\upsilon
       _2^{\infty}\right)\right]\right\},
    \end{align}
   where  \begin{align}
           \upsilon_1^{\infty}=\left(\frac{m_{nd}}{\Omega_{nd}}\right)^{-\frac{\delta}{2}}\frac{\Gamma\left(m_{nd}N+\frac{\delta}{2}\right)}{\delta\Gamma(\delta)\Gamma\left(m_{nd}N\right)}\left(\frac{1}{\zeta}\sqrt{\frac{\frac{\left(\varepsilon_2+\varepsilon_1\right)}{2}\psi}{\rho\beta}}\right)^\delta\; \text{and} \nonumber
       \end{align}
         \begin{align}
           \upsilon_2^{\infty}=\frac{1}{\delta\;\Gamma(\delta)}\left(\frac{1}{\zeta}\sqrt{\frac{\frac{\left(\varepsilon_2+\varepsilon_1\right)}{2}\psi}{\rho\beta}}\right)^\delta.\nonumber
       \end{align}
   \end{proposition}
      \begin{proof}
The derivations of \eqref{out_final_asymptoic} follow the same procedure outlined in Appendix~A up to \eqref{Riemann}. Note that, the lower incomplete gamma function is approximated as $\gamma(a,z) \underset{z \to 0}{\approx} \tfrac{z^a}{a}$~\cite{my_paper}. 
For the term $\upsilon^{\infty}_2$, this approximation directly leads to the asymptotic result. In contrast, for $\upsilon^{\infty}_1$, an additional step is required, where the resulting integral is evaluated to obtain the analytical expression. 
Finally, substituting $\upsilon^{\infty}_1$ and $\upsilon^{\infty}_1$ into \eqref{Riemann} yields the asymptotic BLER expression in \eqref{out_final_asymptoic}.
   \end{proof}
\begin{proposition}
    The asymptotic BLER at the receiving device considering non-uniform reflection coefficient among all the RIS element can be approximated as
    \begin{align}\label{asymp_non_uniform}
       \hat{\varphi}^{\infty} \approx \varpi\sqrt{\Xi}\left(\varepsilon_2-\varepsilon_1\right)\left\{1- \left[\left(1-\nu_1^{\infty}\right) \left(1-\nu
       _2^{\infty}\right)\right]\right\},
    \end{align}
   where \begin{align}
           \nu_1^{\infty}=\left(\frac{m_{nd}}{\Omega_{nd}}\right)^{-\frac{\hat{\delta}}{2}}\frac{\Gamma\left(m_{nd}N+\frac{\hat{\delta}}{2}\right)}{\hat{\delta}\Gamma(\hat{\delta})\Gamma\left(m_{nd}N\right)}\left(\frac{1}{\zeta_1}\sqrt{\frac{\frac{\left(\varepsilon_2+\varepsilon_1\right)}{2}\hat{\psi}}{\rho}}\right)^{\hat{\delta}}\; \text{and}  \nonumber
       \end{align}  \begin{align}
           \nu_2^{\infty}=\frac{1}{\hat{\delta}\;\Gamma(\hat{\delta})}\left(\frac{1}{\zeta_1}\sqrt{\frac{\frac{\left(\varepsilon_2+\varepsilon_1\right)}{2}\hat{\psi}}{\rho}}\right)^{\hat{\delta}}.\nonumber
       \end{align}
   \end{proposition}
    \begin{remark}
       At high SNR, the asymptotic BLER for the uniform and non-uniform reflection coefficient scales as $\varphi^{\infty} \sim \rho^{-{\delta}/2}$ and  $\varphi^{\infty} \sim \rho^{-\hat{\delta}/2}$, which implies a diversity order of ${\delta}/2$ and $\hat{\delta}/2$, respectively.
   \end{remark}
Having established the BLER expressions, we next focus on the goodput. Measured in nats per channel use (npcu), the goodput quantifies the effective transmission rate by accounting for the successful delivery of data frames over the network, including both training and data transmission phases, relative to the total channel uses~\cite{chung_shortpack_Tvt_21mar}. The goodput formulations for scenarios with uniform and non-uniform RIS reflection coefficients are presented in the following corollary.
\begin{corollary}
     The goodput at the receiver considering uniform reflection coefficient  can be evaluated as
     $\mathcal{G} = \Bigl[\left(1-{\frac{1}{\chi}}\right) r\Bigl]\left(1-{\varphi}\right),$
     where $\chi = \Xi+\varrho$ with $\varrho$ denoting the channel uses for training phase. Note that, by substituting $\varphi$ with $\hat{\varphi}$, the goodput expression for the non-uniform case can be directly obtained.
\end{corollary}
\begin{figure*}[t]
    \centering
    \begin{minipage}[b]{0.32\textwidth}
        \centering
        \includegraphics[width=\textwidth]{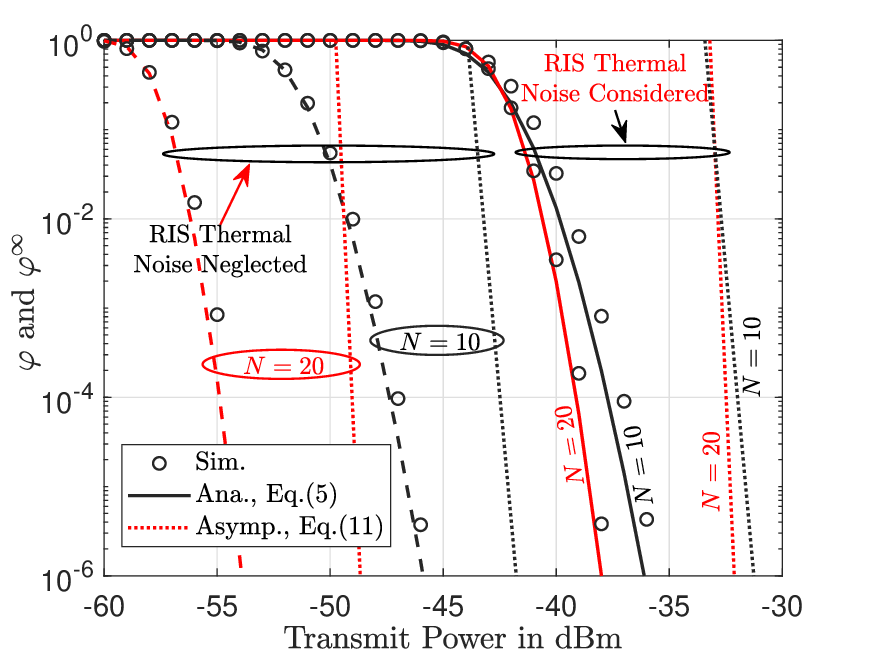}
        \caption{BLER  for uniform $\beta_n$.}
        \label{BLER_unf_fig}
    \end{minipage}
     \begin{minipage}[b]{0.32\textwidth}
        \centering
        \includegraphics[width=\textwidth]{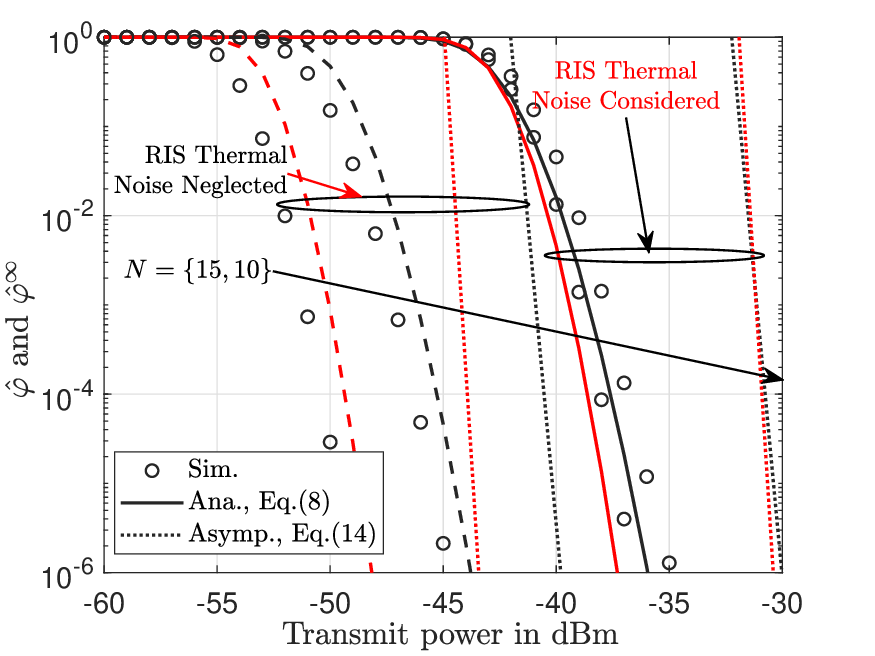}
        \caption{BLER  for non-uniform $\beta_n$.}
        \label{non_unif_fig}
    \end{minipage}
     \centering
    \begin{minipage}[b]{0.32\textwidth}
        \centering
        \includegraphics[width=\textwidth]{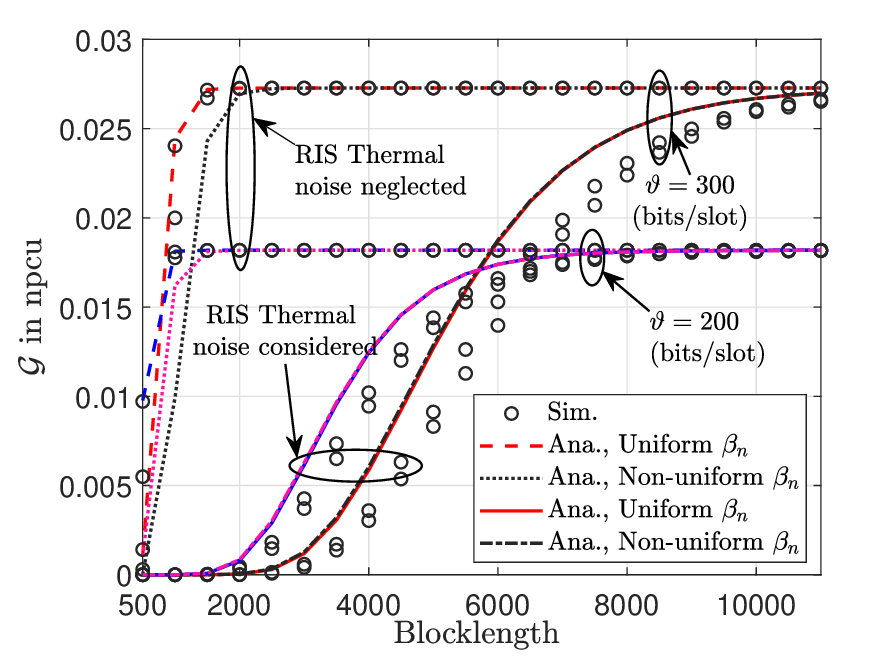}
        \caption{Goodput versus blocklength.}
        \label{good}
    \end{minipage}
    \vspace{-1em}
\end{figure*}
\section{Numerical Results and Discussions}
 
Monte Carlo simulations are conducted in this section to corroborate the accuracy of the derived closed-form expressions in Section~\ref{sec_3}. Unless otherwise noted, the complete set of simulation parameters is given in Table~\ref{t3}. The receiver noise is calculated, using the formulation in Section~\ref{sec:sysmtem}, where the noise figure parameter $\lambda$ is fixed\footnote{Low-noise amplifier designs for 5G new-radio systems achieve $1.3-1.4 $~ dB noise figure value~\cite{lna_5g}, motivating our $3$ dB assumption.} at $3$ dB.

The propagation environment between the transmitter and receiver is modeled through a distance-dependent path-loss of the form $\varsigma/D_i^{-\tau_i}$, where the reference path-loss parameter $\varsigma=1$~m. Here, $D_i$ denotes the link distance, while $\tau_i$, for $i \in \{{bn}, {nd}\}$, represents the path-loss exponents associated with the transmitter-RIS and RIS-receiver links, respectively.
   \begin{table}[t]	\renewcommand{\arraystretch}{1.0}
		\centering
		\caption{ Simulation Parameters.}
		\label{t3}
			\resizebox{\columnwidth}{!}{\begin{tabular}{|l|l|l|l|l|l|}
			\hline
			Parameter         & Value         & Parameter & Value  & Parameter & Value  \\ \hline
			$m_{bn}  $   &     $3 $   	&  $m_{nd}  $   &     $3$ & $\sigma^2_d$   &     $-131.5 $ dB \\ \hline
			
			$D_{bn}$    &    $ 125 $~m &	  $\tau_{bn}$ &     $3.1$ &   $\tau_{nd}$ &     $1.7$  	 \\ \hline	

           $D_{nd}  $   &     $3 $ &	  $B$ &     $10$~MHz &   $\beta, (\text{Uniform}) $ &     $0.9 $  	 \\ \hline

           $\Xi  $   &     $500 $ &	  $\vartheta$ &     $200$ &   $ $ &     $  $  	 \\ \hline	

		\end{tabular}}
        \vspace{-1.2em}
	\end{table}  

 Fig.~\ref{BLER_unf_fig} and Fig.~\ref{non_unif_fig} show the BLER versus transmit power for RIS configurations with uniform and non-uniform reflection coefficients ($\beta_n$), respectively, together with the corresponding asymptotic results. In both cases, results are shown with and without RIS-generated thermal noise ($\sigma_r^2$). For uniform $\beta_n$, the RIS thermal noise values are $\sigma_r^2=-124.4$~dB for $N=10$ and $\sigma_r^2=-121.4$~dB for $N=20$, calculated using the formulation in Section~\ref{sec:sysmtem}. As expected, increasing the transmit power or the number of reflecting elements improves system performance.
A notable observation is the pronounced effect of RIS thermal noise. Ignoring $\sigma_r^2$ consistently yields overly optimistic performance gains, with the mismatch becoming more severe as the number of elements grows when $\sigma^2_r$ is considered. For example, in Fig.~\ref{BLER_unf_fig}, when $N = 10$, achieving a BLER of $10^{-4}$ requires $P \approx -37.8$~dBm when thermal noise is included, compared to $P \approx -47.4$~dBm when it is neglected, corresponding to a gap of nearly $10$~dBm. For $N = 20$, this difference increases to approximately $15$~dBm, demonstrating that the inaccuracy of noise-free assumptions increases with larger RIS elements. 

The same trend becomes more striking when comparing power savings. In the absence of RIS noise, increasing $N$ from $10$ to $20$ reduces the required transmit power from $-47.4$~dBm to $-54.9$~dBm, yielding a gain of approximately $7.5$~dBm. When thermal noise is taken into account, the required power decreases only from $-37.8$~dBm to $-39.1$~dBm, corresponding to a much smaller improvement of about $2.7$~dBm. This contrast reveals that the widely reported energy efficiency gains of RIS can be significantly overestimated if thermal noise is omitted.
A similar behavior is observed in Fig.~\ref{non_unif_fig}. For $N=10$, four elements are set to $\beta_n=0.7$, four to $\beta_n=0.9$, and two to $\beta_n=0.6$, resulting in RIS thermal noise of $-125.1$~dB. For $N=15$, the additional five elements are assigned $\beta_n=0.4$ (two elements) and $\beta_n=0.8$ (three elements) and $\sigma^2_r=-123.6$~dB. Furthermore, comparison of the two figures  for $N=10$ case shows that the uniform $\beta_n$ assumption yields slightly better performance than the non-uniform case; however, the difference is marginal. This suggests that moderate variations in reflection coefficients can be well approximated by their average value.

Finally, the close match between simulation and analytical results confirms the accuracy of the BLER approximations in~\eqref{out_final} and~\eqref{non-uniform} for uniform and non-uniform $\beta_n$, respectively. The asymptotic expressions in~\eqref{out_final_asymptoic} and~\eqref{asymp_non_uniform} also align well with both simulation and analysis as $P \to \infty$, further validating the derived formulations.
An additional insight emerges from the asymptotic analysis, where the influence of RIS thermal noise becomes more pronounced than in the standard BLER curves. Specifically, in Fig.~\ref{BLER_unf_fig}, there exists a transmit power region in which the configuration with $N=10$ outperforms $N=20$. While this crossover is only faintly visible around $P \approx -43$~dBm in the BLER curves, it is more clearly captured by the asymptotic trends. A similar effect appears in Fig.~\ref{non_unif_fig} when $N$ increases from $10$ to $15$. This behavior can be attributed to the fact that increasing the number of RIS elements also amplifies the accumulated thermal noise, which becomes dominant in the low transmit power regime and can offset the expected gains from larger $N$.

Fig.~\ref{good} shows the goodput $(\mathcal{G})$ measured in nats per channel use (npcu) versus blocklength $(\Xi)$ at $P=-52$~dBm for $N=10$, both in the presence and absence of RIS thermal noise. The analytical results closely match simulations, validating Corollary~1. As $\Xi$ increases, $\mathcal{G}$ improves due to reduced BLER, but eventually saturates; around $\Xi \approx 2000$ when $\sigma_r^2=0$ and $\Xi \approx 10,000$ when $\sigma_r^2 \neq 0$ for $\vartheta=300$. A similar trend is observed for $\vartheta=200$, albeit with lower goodput.
Notably, the substantial shift in the saturation point from $\Xi \approx 2000$ to nearly $\Xi \approx 10000$ highlights the significant impact of RIS thermal noise on system performance. Neglecting RIS thermal noise may therefore lead to a severe overestimation of achievable goodput in practical systems. In particular, while a blocklength of approximately $2000$  can be compatible with low-latency transmission requirements, requiring nearly $\Xi\approx 10000$ to achieve saturation can fundamentally undermine the objectives of URLLC applications, where stringent latency and reliability constraints are critical.

 \begin{figure}[t]
    \centering
    \includegraphics[width=0.4\textwidth,clip]
    {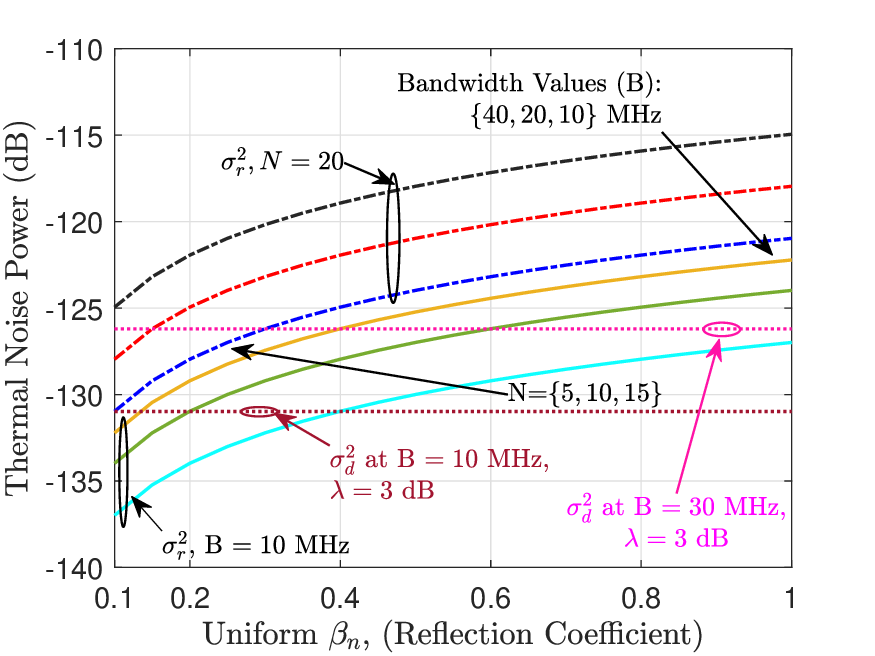}
    \caption{Sensitivity analysis.}
    \label{sense}
   \vspace{-1.2em}
\end{figure}

This behavior reflects a fundamental trade-off: although longer blocklengths improve reliability by reducing BLER, the fixed bits/slot constraint limits the achievable rate gains, resulting in diminishing returns beyond a certain $\Xi$. Finally, the close overlap between uniform and non-uniform $\beta_n$ curves indicates that their performance difference is marginal, consistent with the observations in Fig.~\ref{BLER_unf_fig} and Fig.~\ref{non_unif_fig}.

Fig.~\ref{sense} shows the RIS thermal noise power as a function of the uniform reflection coefficient $\beta_n$ across all the RIS elements for different $N$ and bandwidth $B$. As expected, smaller $\beta_n$ generates less noise. From this it can be interpreted that $\beta_n$ represents both the RIS reflectivity and an indicator of signal degradation. Since receiver noise depends only on $B$, whereas RIS-induced noise scales with both $B$ and $N$, the figure shows that for $B=10$~MHz and $N=5$, $\sigma_d^2$ dominates $\sigma_r^2$ when $\beta_n \leq 0.4$, allowing $\sigma_r^2$ to be neglected in this region. For larger $\beta_n$, however, $\sigma_r^2$ becomes significant, and its impact grows with increasing $B$ and $N$. This highlights the tradeoff between RIS size and bandwidth on one hand, and the resulting thermal noise on the other. From a design perspective, RIS thermal noise is negligible when $B$ and $N$ are small, but emerges as a dominant impairment as they increase.

\section{Conclusion and Future research directions}
In this work, we studied the impact of thermal noise generated by a passive RIS. Closed-form expressions for the BLER, asymptotic BLER, and goodput were derived and validated through simulations for both uniform and non-uniform RIS reflection coefficients. The results reveal that neglecting this effect leads to increasingly optimistic performance predictions as the number of RIS elements grows, significantly overstating both reliability and energy efficiency gains. Furthermore, the analysis shows that increasing the number of elements does not always improve performance, particularly in the low transmit power regime where the accumulated RIS noise becomes dominant. These findings emphasize the need to account for RIS thermal noise for realistic performance evaluation, especially in large-scale deployments.
{\appendices
\section{}
In order to evaluate the expression for BLER, we first apply piecewise linear approximation on \eqref{Qfunc} as $Q\left(\frac{C(\gamma_d)-r}{\sqrt{V(\gamma_d)/\Xi}}\right)\simeq \Psi\left(\gamma_d\right)$, which can be expressed as 
 \begin{equation}\label{qexpand}
\Psi\left(\gamma_d\right)=\begin{cases}
1, &    {\gamma_d}\leq \varepsilon_1, \\
0.5- \varpi\sqrt{\Xi}\left({\gamma_d}-\Lambda\right), & \varepsilon_1\leq  {\gamma_d} \leq \varepsilon_2\\
0, & {\gamma_d}\geq \varepsilon_2. \end{cases}
\end{equation}
Taking the mean of $\Psi_d\left( \gamma_d\right)$, characterized by the PDF of $\gamma_d$ and then applying integration by parts   we simplify \eqref{qexpand} as $\varphi \approx \varpi\sqrt{\Xi}\int\limits^{\varepsilon_2}_{\varepsilon_1}F{_{\gamma_d}}(y) dy$
The integral involved here solved using the Riemann integral approach as$\int\limits^a_b g(y)dy\approx(a-b)g((a+b)/2)$, we obtain
\begin{align}\label{Riemann}
     \varphi \approx \varpi\sqrt{\Xi}\left(\varepsilon_2-\varepsilon_1\right)F{_{\gamma_d}}\left(\left(\varepsilon_1+\varepsilon_2\right)/2\right),
\end{align} where $F{_{\gamma_d}}$ is the CDF of \eqref{SINR}. In order to evaluate \eqref{Riemann}, we first calculate the CDF of ${\gamma_d}$. To make it analytically tractable, we  approximate \eqref{SINR}, using lower bound and upper bound, which can be expressed as 
$   \gamma^{LB} =\frac{1}{2}\text{min}\left(\frac{{\rho} \left|\mathbf{h}^{\mathrm{H}}_{d} \boldsymbol{\Theta} \mathbf{h}_{b}\right|^2}{||\mathbf{h}^{\mathrm{H}}_{d}||^2 \psi}, {\rho} \left|\mathbf{h}^{\mathrm{H}}_{d} \boldsymbol{\Theta} \mathbf{h}_{b}\right|^2 \right),$
and
$    \gamma^{UB} =\text{min}\left(\frac{{\rho} \left|\mathbf{h}^{\mathrm{H}}_{d} \boldsymbol{\Theta} \mathbf{h}_{b}\right|^2}{||\mathbf{h}^{\mathrm{H}}_{d} ||^2 \psi}, {\rho} \left|\mathbf{h}^{\mathrm{H}}_{d} \boldsymbol{\Theta} \mathbf{h}_{b}\right|^2 \right)$. Therefore, $F_{\gamma_d}(y)= 1 -\left[ \left(1 - {\upsilon_1}\right) \left( 1 - {\upsilon_2}\right) \right]$ such that $\upsilon_1=\left\{\mathrm{Pr} \left( \frac{{\rho} \left|\mathbf{h}^{\mathrm{H}}_{d} \boldsymbol{\Theta} \mathbf{h}_{b}\right|^2}{||\mathbf{h}^{\mathrm{H}}_{d} ||^2 \psi} \right) < y\right\}$ and $\upsilon_2=\left\{\mathrm{Pr} \left( {\rho} \left|\mathbf{h}^{\mathrm{H}}_{d} \boldsymbol{\Theta} \mathbf{h}_{b}\right|^2 \right) < y\right\}$.
To evaluate $\upsilon_1$ and $\upsilon_2$ for uniform $\beta_n$, let $X \triangleq \left(\sum\limits^{N}_{n=1}|h_{bn}||h_{nd}|\right)^2$, which can be approximated by a Gamma distribution with CDF~\cite{dEEPAK_TWC_25}, as $F_X(x) = \frac{1}{\Gamma(\delta)} \gamma\left(\delta, \frac{\sqrt{x}}{\zeta}\right).$
Also, define $f_Z(z)= \left(\frac{m_{nd}}{\Omega_{nd}}\right)^{m_{nd}N} \frac{z^{m_{nd}N-1}}{\Gamma(m_{nd}N)} \exp\left(-\frac{m_{nd}}{\Omega_{nd}}z\right)$.
Now, $\upsilon_1$ can be expressed as
\begin{align}\label{int_upsilon}
\upsilon_1 = \int^{\infty}_{0} f_Z(z)F_X\left(\frac{\psi y z}{\rho\beta}\right) dz.
\end{align}
Substituting the CDF and PDF obtained into \eqref{int_upsilon}
and doing some mathematical simplifications, we obtain 
\begin{align}\label{step2}
   & \upsilon_1 =1- \sum\limits^{\deltaInSqcup \!- 1}_{i=0} \frac{1}{i!\sqrt{\pi}} \left(\frac{m_{nd}}{\Omega_{nd}}\right)^{-m_{nd}N}\!\left(\frac{1}{\zeta}\sqrt{\frac{y\psi}{\rho\beta}}\right)^i \frac{1}{\Gamma(m_{nd}N)}\nonumber\\
        &\times\int^{\infty}_{0}\!\!\! z^{{m_{nd}N}+\frac{i}{2}-1}\exp \left(\frac{-m_{nd}}{\Omega_{nd}}z\right) G_{0, 2}^{2,0} \left( \frac{\psi y z}{4\rho\beta\zeta^2} \Bigg| \begin{array}{c}0 \\0, \frac{1}{2} \end{array} \right)dz.
\end{align}
 Finally  by using~\cite[Eq. $07.34.21.0088.01$]{wolfram}, we obtain $\upsilon_1$. The derivation of $\upsilon_2$ is straightforward and omitted for brevity.
}
\bibliographystyle{IEEEtran_renamed}
\bibliography{ref}
\end{document}